\documentclass[12pt]{article}

\usepackage{fullpage}
\usepackage{amsmath}
\usepackage{xspace}
\usepackage{amssymb}
\usepackage{epsfig}
\usepackage[usenames,dvipsnames]{color}

\newtheorem{Thm}{Theorem}
\newtheorem{Lem}[Thm]{Lemma}
\newtheorem{Cor}[Thm]{Corollary}

\newtheorem{Def}{Definition}

\newcommand{\qed}{\hfill{$\rule{6pt}{6pt}$}} 
\newenvironment{proof}{\noindent{\bf Proof}:}{\qed}

\newcommand\supp{\mbox{\tt {supp}}\xspace}

\mathchardef\mhyphen="2D

\newcommand{\ket}[1]{| #1 \rangle}
\newcommand{\bra}[1]{\langle #1 |}

\newcommand\qip[2]{\langle #1 | #2 \rangle}

\newcommand{\av}{{\mathbb{E}}}

\newcommand\mbR{\mbox{$\mathbb{R}$}}

\newcommand {\ie} {\textit{i.e.}\xspace}

\begin{document}
\title{Quantum game players can have advantage without discord}
\author{Zhaohui Wei\thanks{School of Physics and Mathematical Sciences,
Nanyang Technological University and Centre for Quantum
Technologies, Singapore.} \ and \ Shengyu Zhang\thanks{Department of
Computer Science and Engineering, The Chinese University of Hong
Kong, Hong Kong.}}
\date{}
\maketitle

\begin{abstract}
The last two decades have witnessed a rapid development of quantum information processing, a new paradigm which studies the power and limit of ``quantum advantages'' in various information processing tasks. Problems such as when 
quantum advantage exists, and if existing, how much it could be, are
at a central position of these studies. In a broad class of
scenarios, there are, implicitly or explicitly, at least two parties
involved, who share a state, and the correlation in this shared
state is the key factor to the efficiency under concern. In these
scenarios, the shared \emph{entanglement} or \emph{discord} is
usually what accounts for quantum advantage. In this paper, we
examine a fundamental problem of this nature from the perspective of
game theory, a branch of applied mathematics studying selfish
behaviors of two or more players. We exhibit a natural zero-sum
game, in which the chance for any player to win the game depends
only on the ending correlation. We show that in a certain classical
equilibrium, a situation in which no player can further increase her
payoff by any local classical operation, whoever first uses a
quantum computer has a big advantage over its classical opponent.
The equilibrium is fair to both players and, as a shared
correlation, it does not contain any discord, yet a quantum
advantage still exists. This indicates that at least in game theory,
the previous notion of discord as a measure of non-classical
correlation needs to be reexamined, when there are two players with
different objectives.
\end{abstract}

\section{Introduction}
Quantum computers have exhibited tremendous power in algorithmic, cryptographic, information theoretic, and many other information processing tasks, compared with their classical counterparts. Meanwhile, for a large number of problems, quantum computers are not able to offer much advantage over classical ones. When and why quantum computers are more powerful are always at a central position in studies on quantum computation and quantum information processing. A particularly interesting class of scenarios is when there are, implicitly or explicitly, at least two parties involved who share a state, the correlation in this state is the key factor. What accounts for the quantum advantage is often \emph{entanglement}, one of the most distinctive characters of quantum information. 
Indeed, it has been showed that a quantum algorithm with only slight entanglement can be simulated efficiently by a classical computer \cite{Vidal03}. In certain potential applications of quantum algorithms, it is also shown that entangled measurement is necessary for the existence of efficient quantum algorithms \cite{HMR+10}.

Recently people started to realize that entanglement is not always a necessary resource needed for generating quantum correlations. It has been found that \emph{discord}, another unique character of quantum states, also plays an important role in quantum information processing \cite{OZ02}. Discord is a relaxed version of entanglement---states with positive entanglement must also have positive discord, but there are states with positive discord but zero entanglement. People has discovered cases where quantum speed-up exists without entanglement involved, and discord is considered to be responsible for the quantum advantage \cite{DSC08}. Till today, discord is widely considered as necessary for the existence of quantum advantages.

In this paper, we reexamine this notion from the perspective of game
theory \cite{OR94}. Game theory studies the situation in which there
are two or more players with possibly different goals. There are two
broad classes of games, one is strategic-form (or normal-form)
games, in which all players make their choice simultaneously; a
typical example is Rock-Paper-Scissors. The other class is
extensive-form games, in which players make their moves in turn; a
typical example is chess.

The research on quantum games began about one decade ago, starting with two pioneering papers.\footnote{Note that there is also a class of ``nonlocal games'', such as CHSH or GHZ games \cite{BCMdW10}, where all the players have the \emph{same} objective. But general game theory focuses more on situation that the players have \emph{different} objective functions, and the players are selfish, each aiming to optimize her own objective function only.} 
The first one \cite{EWL99} aimed to quantize a specific
strategic-form game called \emph{Prisoners' Dilemma} \cite{EWL99},
and it unleashed a long sequence of follow-up works in the same
model. Despite the rapid growth of literature, controversy also
largely exists \cite{BH01,vEP02,CT06}, which questioned the meaning
of the claimed quantum solution, the ad hoc assumptions in the
model, and the inconsistency with standard settings of classical
strategic games. Recently a new model was proposed for quantizing
general strategic-form games \cite{Zha12}. Compared with
\cite{EWL99}, the new model corresponds to the classical games more
precisely, and has rich mathematical structures and game-theoretic
questions; also see later theoretical developments
\cite{KZ12,WZ13,JSWZ13,PKL+15}.

Back to the early stage of the development of quantum game theory,
the other pioneering paper was \cite{Mey99}, which demonstrated the
power of using quantum strategies in an extensive-form game. More
specifically, Meyer considered the quantum version of the classical
Penny Matching game. The basic setting is as follows. There are two
players, and each has two possible actions on one bit: Flip it or
not. Starting with the bit being 0, Player 1 first takes an action,
and then Player 2 takes an action, and finally Player 1 takes
another action, and the game is finished. If the bit is finally 0,
then Player 1 wins; otherwise Player 2 wins. It is not hard to see
that if Player 2 flips the bit with half probability, then no matter
what Player 1 does, each player wins the game with half probability.
Now consider the following change of setting: The bit becomes a
qubit; the first player uses a quantum computer in the sense that
she can perform any quantum admissible operation on the bit; the
second player uses a classical computer in the sense that she can
perform either Identity or the flip operation $\begin{bmatrix} 0 & 1
\\ 1 & 0\end{bmatrix}$. In this new setting, Player 1 can win the
game with certainty! Her winning strategy is simple: she first
applies a Hadamard gate to change the state to $\ket{+} =
(\ket{0}+\ket{1})/\sqrt{2}$, and then no matter whether Player 2
applies the flip operation or not, the state remains the same
$\ket{+}$, thus in the third step Player 1 can simply apply a
Hadamard gate again to rotate the state back to $\ket{0}$. This
shows that a player using a quantum computer can have big advantage
over one using a classical computer.

Despite a very interesting phenomena it exhibits, the quantum
advantage is not the most convincing due to a fairness issue. After
all, the quantum player takes two actions and the classical player
takes just one. And the order of ``Player 1 $\rightarrow$ Player 2
$\rightarrow$ Player 1'' is also crucial for the quantum advantage.
One remedy is to consider normal-form games, in which the players
give their strategies \emph{simultaneously}, thus there is no longer
the issue of the action order. Taking the model in \cite{Zha12}, two
players play a complete-information normal-form game, with a
starting state $\rho$ in systems $(A_1,A_2)$, and $A_i$ being given
to Player $i$. A classical player can only measure her part of the
state in the computational basis, followed by whatever classical
operation $C$ (on the computational basis). In previous works
\cite{EWL99,Mey99,ZWC+12} the classical player is usually assumed to
be able to apply any classical operation on computational basis
(such as $X$-gate), followed by a measurement in the computational
basis. A classical operations there is implicitly assumed to be
unitary, so the operation in the matrix form is a permutation
matrix. Here we allow classical player to measure first and then
perform any classical operation, which gives her more power since
the second-step classical operation does need to be unitary. Indeed,
in Meyer's Penny Matching game, in the second step Player 2 could
measure the state first and then randomly set it to be $\ket{0}$ or
$\ket{1}$ each with half probability. Then in the third step, Player
1's Hadamard gate will change the state to $\ket{+}$ or $\ket{-}$,
in either case, Player 1 could win with only half probability.

Even if we now enlarge the space of possible operations of the
classical player, we will show examples where the quantum player has
advantage of winning the game. Furthermore, the examples have the
following nice properties respecting the fairness of the game:
\begin{enumerate}
    \item If both players are classical, then both get expected payoff 0, and $\rho$ is a correlated equilibrium in the sense that any classical operation $C$ by one player cannot increase her expected payoff.
    \item Suppose that one player remains classical and the other player uses a quantum computer. To illustrate the power of using quantum strategies, we cut the classical player some slack as follows. The classical player can (1) pick one subsystem, $A_1$ or $A_2$, of $\rho$, leaving the other subsystem to the quantum player, and (2) ``take side" by picking one of the two payoff matrices, leaving the other to the quantum player.
\end{enumerate}

Examples were found that even with the advantage of taking side and taking part of the shared state, the classical player still has a disadvantage compared to the quantum player. Consider the canonical $2\times 2$ zero-sum game with the payoff matrices being
\begin{equation}\label{eq:Penny game}
U_1=
\begin{pmatrix}
1 & -1\\
-1 & 1
\end{pmatrix}
\quad \text{ and } \quad
U_2=
\begin{pmatrix}
-1 & 1\\
1 & -1
\end{pmatrix}.
\end{equation}
\paragraph{Quantum game with entanglement}
Each player $i$ owns a 2-dimensional Hilbert space, and they share the quantum state
\begin{equation}
\ket{\psi} = \frac{1}{\sqrt{2}}(\ket{+0} + \ket{-1}) = \frac{1}{\sqrt{2}}(\ket{0+}+\ket{1-}),
\end{equation}
where $\ket{+}=\frac{1}{\sqrt{2}}(\ket{0}+\ket{1})$, and $\ket{-}=\frac{1}{\sqrt{2}}(\ket{0}-\ket{1})$.
It is not difficult to verify that if both players measure their parts in the computational basis, then each gets payoff 1 and $-1$ with equal probability, resulting an average payoff of zero for both players. This is a correlated equilibrium for classical operations.

Now suppose that Player 1 employs a quantum computer. Since the state is symmetric, it does not matter which part Player 2, the classical player, chooses. Let us assume that Player 2 chooses part 2, and the payoff matrix $U_2$. Then Player 1 can apply the Hadamard transformation on her qubit, followed by the measurement in computational basis. The state immediately before the measurement is
$\ket{\psi'} = (\ket{00} + \ket{11})/\sqrt{2}$. Therefore the measurement in the computational basis gives Player 1 and Player 2 payoff 1 and $-1$, respectively, with certainty. In other words, Player 1 wins with certainty, whereas she could only win with half probability when using a classical computer.

In this example where the quantum player has an advantage, the state shared by players is highly entangled, which motivates the following natural question: Is entanglement necessary for quantum advantage in the game? It turns out that the answer is no. Consider the example below.

\paragraph{Quantum game with discord} The payoff matrices are the same as before, but the quantum state shared by players is the following.
\begin{equation}
\rho =
\frac{1}{4}(\ket{+}\bra{+}\otimes\ket{0}\bra{0} + \ket{0}\bra{0}\otimes\ket{+}\bra{+} + \ket{-}\bra{-}\otimes \ket{1}\bra{1} + \ket{1}\bra{1}\otimes \ket{-}\bra{-}).
\end{equation}
This state is separable and thus does not have any entanglement. It can be checked that if the players measure this state in computational basis, the probability of getting each of the four possible outcomes is 1/4. Thus the overall payoff of each player is zero, and it can be verified that it is a classical correlated equilibrium.

In the quantum setting, again without loss of generality assume that the classical computer picks the second part of $\rho$ and the second payoff matrix. The quantum player can again perform a Hadamard operation on her system, resulting in a new state
\begin{equation}\label{eq:noEntan}
\rho' =
\frac{1}{4}(|0\rangle\langle0|\otimes|0\rangle\langle0|+|+\rangle\langle+|\otimes|+\rangle\langle+|+|1\rangle\langle1|\otimes|1\rangle\langle1|+|-\rangle\langle-|\otimes|-\rangle\langle-|).
\end{equation}
Measuring the new state, the quantum player gets state $\ket{00}$, $\ket{01}$, $\ket{10}$, $\ket{11}$ with probability
$3/8$, $1/8$, $1/8$, $3/8$ respectively. As a result, her winning probability
increases from 1/2 to 3/4; in other words, she gets an expected payoff of 1/2.

Note that the quantum state in Eq.\eqref{eq:noEntan} is separable,
and there is no any entanglement, but the quantum player still gets
a quantum advantage. Thus, entanglement is not necessary for quantum
advantage to exist in this game. Note that, however, the state in
Eq.\eqref{eq:noEntan} has a positive discord. As we have mentioned, it was known that in some scenarios, it is discord, rather than
entanglement, that produces non-classical correlations. So the
above example confirms this traditional notion in the new
game-theoretic setting.

These two examples were also experimentally verified recently \cite{ZWC+12}. The present paper makes further studies on the foregoing notion by asking the following fundamental question.

\begin{quote}
    \emph{Is discord necessary for quantum advantage to exist in games where players share a symmetric state?}
\end{quote}
It is tempting to conjecture that the answer is Yes. In the rest of the paper, we will show that, first, discord is indeed necessary for any quantum advantage to exist in a 2-player games where each player has $n=2$ strategies. We will then show that when $n\geq 3$, however, there are games where the quantum player has a positive advantage even when the shared symmetric state has zero discord.

\section{Preliminaries}
Suppose that in a classical game there are $k$ players, labeled by
$\{1,2,\ldots,k\}$. Each player $i$ has a set $S_i$ of strategies.
To play the game, each player $i$ selects a strategy $s_i$ from
$S_i$. We use $s=(s_1,\ldots, s_k)$ to denote the \emph{joint
strategy} selected by the players and $S= S_1 \times \ldots \times
S_k$ to denote the set of all possible joint strategies. Each player
$i$ has a utility function $u_i: S \rightarrow \mbR$, specifying the
\emph{payoff} or \emph{utility} $u_i(s)$ of Player $i$ on the joint
strategy $s$. For simplicity of notation, we use subscript $-i$ to
denote the set $[k]-\{i\}$, so $s_{-i}$ is $(s_1, \ldots, s_{i-1},
s_{i+1}, \ldots, s_k)$, and similarly for $S_{-i}$, $p_{-i}$, etc.
In this paper, we will mainly consider 2-player games.

Nash equilibrium is a fundamental solution concept in game theory. Roughly, it says that in a joint strategy, no player can gain more by changing her strategy, provided that all other players keep their current strategies unchanged. The precise definition is as follows.
\begin{Def}
A \emph{pure Nash equilibrium} is a joint strategy $s = (s_1, \ldots ,s_k) \in S$ satisfying that
\begin{align}
    u_i(s_i,s_{-i}) \geq  u_i(s_i',s_{-i})
\end{align}
for all $i\in [k]$ and all $s'_i\in S_i$.
\end{Def}
Pure Nash equilibria can be generalized by allowing each player to independently select her strategy according to some probability distribution, leading to the following concept of \emph{mixed Nash equilibrium}.

\begin{Def}
A \emph{(mixed) Nash equilibrium (NE)} is a product probability distribution $p = p_1 \times \ldots \times p_k$, where each $p_i$ is a probability distributions over $S_i$, satisfying that
\begin{align}
    \sum_{s_{-i}} p_{-i}(s_{-i}) u_i(s_i,s_{-i}) \geq  \sum_{s_{-i}} p_{-i}(s_{-i}) u_i(s_i',s_{-i}),
\end{align}
for all $i\in [k]$, and all $s_i, s'_i\in S_i$ with $p_i(s_i)>0$.
\end{Def}

A fundamental fact proved by Nash \cite{Nas51} is that every game with a finite number of players and a finite set of strategies for each player has at least one mixed Nash equilibrium.

There are various further extensions of mixed Nash equilibria. Aumann \cite{Aum74} introduced a relaxation called \emph{correlated equilibrium}. This notion assumes an external party, called Referee, to draw a joint strategy $s = (s_1, ..., s_k)$ from some probability distribution $p$ over $S$, possibly correlated in an arbitrary way, and to suggest $s_i$ to Player $i$. Note that Player $i$ only sees $s_i$, thus the rest strategy $s_{-i}$ is a random variable over $S_{-i}$ distributed according to the conditional distribution $p|_{s_i}$, the distribution $p$ conditioned on the $i$-th part being $s_i$. Now $p$ is a correlated equilibrium if any Player $i$, upon receiving a suggested strategy $s_i$, has no incentive to change her strategy to a different $s_i' \in S_i$, assuming that all other players stick to their received suggestion $s_{-i}$.

\begin{Def} \label{thm:CE}
A \emph{correlated equilibrium (CE)} is a probability distribution $p$ over $S$ satisfying that
\begin{align}
    \sum_{s_{-i}} p(s_i,s_{-i}) u_i(s_i,s_{-i}) \geq  \sum_{s_{-i}} p(s_i,s_{-i}) u_i(s_i',s_{-i}),
\end{align}
for all $i\in [k]$, and all $s_i, s'_i\in S_i$.
\end{Def}
The above statement can also be restated as
\begin{equation}
    \av_{s_{-i} \gets \mu|s_i}[u_i(s_i,s_{-i})] \ge \av_{s_{-i} \gets \mu|s_i}[u_i(s_i',s_{-i})].
\end{equation}
where $\mu|s_i$ is the distribution $\mu$ conditioned on the $i$-th component being $s_i$.
Notice that a classical correlated equilibrium $p$ is a classical Nash equilibrium if $p$ is a product distribution.

Correlated equilibria captures natural games such as the Traffic Light and the Battle of the Sexes (\cite{VNRT07}, Chapter 1). The set of CE also has good mathematical properties such as being convex (with Nash equilibria being some of the vertices of the polytope). Algorithmically, it is computationally benign for finding the best CE, measured by any linear function of payoffs, simply by solving a linear program (of polynomial size for games of constant players). A natural learning dynamics also leads to an approximate CE (\hspace{-.08em}\cite{VNRT07}, Chapter 4) which we will define next, and all CE in a graphical game with $n$ players and with $\log(n)$ degree can be found in polynomial time (\hspace{-.08em}\cite{VNRT07}, Chapter 7).

\section{Quantum game without discord}

In this section, we will address the question proposed at the end of the first section. 
Suppose that a game has two players and both of them have $n$
strategies. In other words, each player holds an $n$-dimensional
quantum system. Recall that we also require the shared quantum state $\rho\in H\otimes H$ be symmetric, so that swapping the two systems does not
change the state. It is not hard to derive from the general criteria
of zero-discord state \cite{DVB10} that these quantum states $\rho$ have the form of
\begin{equation}\label{eq:def of rho}
\rho=\sum_{i,j=0}^{n-1}p(i,j)|\psi_i\rangle\langle\psi_i|\otimes|\psi_j\rangle\langle\psi_j|,
\end{equation}
where $\{|\psi_i\rangle\}$ is a set of orthogonal basis of the $n$-dimensional Hilbert space $H$, and $P=[p(i,j)]_{ij}\in {\mbR}_+^{n\times n}$ is a
symmetric matrix with nonnegative entries satisfying that $\sum_{ij}
p(i,j) = 1$. (In general, we use the upper case letter $P$ to denote
the matrix and the lower case letter $p$ to denote the corresponding
two-variate distribution $p(i,j)$.) We sometimes also write the
state as
\begin{equation}
\rho=\sum_{i}p_1(i)\ket{\psi_i}\bra{\psi_i}\otimes \sigma_i
\end{equation}
where $p_1(i) = \sum_j p(i,j)$ is the marginal distribution on the
first system, and $\sigma_i = \sum_j \frac{p(i,j)}{p_1(i)}
\ket{\psi_j}\bra{\psi_j}$ (if $p_1(i) = 0$ then let $\sigma_i =
\ket{0}\bra{0}$).

Consider the following game as a natural extension of the Penny Matching game in Section 1. The payoff matrices are

\begin{equation}\label{eq:dice}
 U_1 = nI - J \quad and \quad U_2 = -U_1,
\end{equation}
where $J$ is the all-one matrix. Intuitively, whoever takes the first matrix bets that the two $n$-sided dice give the same side, and the other player bets that the two dice give different sides.
We first show that there is a unique correlated equilibrium in the game.
\begin{Lem}\label{lem:ce}
The game given by Eq.\eqref{eq:dice} has only one classical correlated equilibrium $Q = J/n^2$.
\end{Lem}
\begin{proof}
According to the definition of correlated equilibrium, if a distribution $q$ on $[n]\times [n]$ is a
classical correlated equilibrium, then the following relationships hold:
\begin{equation}\label{eq:forCE1}
    \sum_{j}q(i,j)U_1(i,j)\geq \sum_{j}q(i,j)U_1(i',j), \qquad \forall i,i'\in\{0,1,...,n-1\},
\end{equation}
and
\begin{equation}\label{eq:forCE2}
    \sum_{i}q(i,j)U_2(i,j)\geq \sum_{i}q(i,j)U_2(i,j'), \qquad \forall j,j'\in\{0,1,...,n-1\}.
\end{equation}
Plugging the definition of $U_1$ and $U_2$ into the above inequalities, one can verify that $Q=J/n^2$ is the only solution.
\end{proof}

Recall that $\rho=\sum_{i}p_1(i)\ket{\psi_i}\bra{\psi_i}\otimes \sigma_i$. Since $\rho$ is symmetric, it does not matter which part the classical player, Player 2, chooses to hold. For the convenience of discussions, let us assume that the classical player takes the second part. We use $\supp(p)$ to denote the support of a distribution $p$, \ie the set of elements with non-zero probability. 
The next lemma gives a sufficient and necessary condition for the existence of quantum advantage. 
\begin{Lem}\label{lem:advan}
Suppose that measuring the state $\rho$ gives a classical correlated equilibrium for the game given in Eq.\eqref{eq:dice}. Then Player 1 (who is quantum) does not have any
advantage if and only if
\begin{equation}\label{eq:lem2}
\langle i|\sigma_j|i\rangle=1/n, \ \ \ \forall i\in \{0,1,...,n-1\}
\text{ and } j\in \text{\supp}(p_1).
\end{equation}
\end{Lem}
\begin{proof}
\emph{``Only if"}: Assume that Player 1 first measures her part in
the orthonormal basis $\{|\psi_i\rangle\}$. Note that this does not affect the state. If outcome $j$
occurs, then Player 1 knows that the state of Player 2 is
$\sigma_j$. We consider which utility matrix in Eq.\eqref{eq:dice} Player 1 has. In the first case, Player 1 takes the utility matrix $U_1$. It is not hard to see that her optimal strategy is to replace
her part $\ket{\psi_j}$ by $\ket{i}$, where $i$ is a maximizer of
$\max_i \langle i|\sigma_j|i\rangle$. Thus Player 1 has a strict
positive advantage if and only if there is some $i$ and $j$, where
$j\in \supp(p_1)$, with $\bra{i}\sigma_j \ket{i} > 1/n$, which is
equivalent to saying that there is some $i$ and $j\in \supp(p_1)$
with $\bra{i}\sigma_j \ket{i} \neq 1/n$.

Similarly, if Player 1 takes the utility matrix $U_2$, then her
optimal strategy is to replace $\ket{\psi_j}$ with $\ket{i}$, where
$i$ is a minimizer of $\min_i \langle i|\sigma_j|i\rangle$. Thus
Player 1 has a strict positive advantage if and only if there is
some $i$ and $j$ with $\bra{i}\sigma_j \ket{i} < 1/n$, which is
again equivalent to saying that there is some $i$ and $j$ with
$\bra{i}\sigma_j \ket{i} \neq 1/n$.

\emph{``If"}: Player 2 measures her part in the computational
basis, yielding the state
\begin{equation*}
\frac{1}{n}\sum_{i,j} p_1(j) \ket{\psi_j}\bra{\psi_j}\otimes \ket{i}\bra{i}.
\end{equation*}
Now whatever quantum operation Player 1 applies, the probability of
observing the same bits (\ie the state after the measurement is
$\ket{ii}$ for some $i$) is $1/n$, with the expected payoff of 0
for both players.
\end{proof}


\medskip
Though the above lemma gives a sufficient and necessary condition, it is still not always clear whether quantum advantage could exist for any symmetric state $\rho$ with zero discord. Next we will further the study by considering a related matrix $M \in \mbR_+^{n\times n}$, whose $(i,j)$-th entry is defined to be
\begin{equation}\label{eq:M}
    M (i,j) = |\qip{i}{\psi_j}|^2.
\end{equation}

It turns out that the rank of $M$ is an important criteria to our question. In the rest of this section, we will consider two cases, depending on whether $M$ is full rank or not.

\subsection{Case 1: $M$ is full-rank}
We will first show that if $M$ is full-rank, then the quantum player cannot have any advantage.
\begin{Thm}
Suppose that the two players of the game Eq.\eqref{eq:dice} share a symmetric state $\rho$, measuring which gives a classical correlated equilibrium. Then Player 1 (who is quantum) does not have any advantage if $M$ in Eq.\eqref{eq:M} is
full-rank.
\end{Thm}
\vspace{0.1in}
\begin{proof}
By Lemma \ref{lem:ce}, for any $0\leq k,j\leq n-1$ we have
\[
\sum_{i=0}^{n-1}p_1(i)|\langle k|\psi_i\rangle|^2\cdot\langle
j|\sigma_i|j\rangle=\frac{1}{n^2}
\]
Summing over $j$, we obtain another equality
\[
\sum_{i=0}^{n-1}p_1(i)|\langle k|\psi_i\rangle|^2=\frac{1}{n}.
\]
Combining these two equalities, we have
\[
\sum_{i=0}^{n-1}|\langle k|\psi_i\rangle|^2\cdot p_1(i)\left(\langle
j|\sigma_i|j\rangle-\frac{1}{n}\right)=0.
\]
Define a matrix $A=[a(i,j)]_{ij}\in {\mbR}^{n\times n}$ by
$a(i,j)=p_1(i)\left(\langle j|\sigma_i|j\rangle-\frac{1}{n}\right)$.
Then the above equality is just $\sum_i M(k,i)a(i,j) = 0$ for all
$k,j$. In other words, we have
\begin{equation}
M\cdot A=0.
\end{equation}
Since the matrix $M$ is assumed to be full-rank, we have $A = M^{-1}
0 = 0$. The conclusion thus follows by Lemma \ref{lem:advan}.
\end{proof}

Two corollaries are in order. First, note that $M$ is full-rank for a generic orthogonal basis $\{|\psi_i\rangle\}$, it is generically true that no discord implies no quantum advantage.
\begin{Cor}
If a set of orthonormal basis $\{\ket{\psi_i}\}$ is picked uniformly at random, then with probability 1, the quantum player does not have any advantage.
\end{Cor}

The second corollary considers the case of $n=2$, which is settled
by the above theorem completely. Indeed, when $n=2$, the rank of $M$
is either 1 or 2. The rank-2 case is handled by the above theorem.
If the rank is 1, it is not hard to see that the only possible $M$
is $M = \begin{bmatrix} 1/2 & 1/2 \\ 1/2 & 1/2 \end{bmatrix}$. In
this case, for any $i$ and any $k$ it holds that
\[
\bra{k}\sigma_i \ket{k} = \bra{k} \Big(\sum_j p(j|i) \ket{\psi_j} \bra{\psi_j}\Big)  \ket{k} = \sum_j p(j|i) |\qip{k}{\psi_j}|^2 = \frac{1}{2} \sum_j p(j|i) = \frac{1}{2}.
\]
Applying Lemma \ref{lem:advan}, we thus get the following corollary.
\begin{Cor}
There is no quantum advantage for the game defined in Eq.\eqref{eq:Penny game} on any symmetric state $\rho$ with zero discord.
\end{Cor}

\subsection{Case 2: $M$ is not full rank}
Somewhat surprisingly, the quantum player \emph{can} have an advantage when $M$ is not full-rank. In this section we exhibit a counterexample for $n=3$. In this case, recall that the payoff matrices are
\begin{equation}
U_1=
\begin{pmatrix}
2 & -1 & -1\\
-1 & 2 & -1\\
-1 & -1 & 2
\end{pmatrix}
\quad \text{ and } \quad
U_2=
\begin{pmatrix}
-2 & 1 & 1\\
1 & -2 & 1\\
1 & 1 & -2
\end{pmatrix}.
\end{equation}

We consider the following quantum state,
\begin{equation}
\rho=\sum_{i,j=0}^{2}p(i,j)|\psi_i\rangle\langle\psi_i|\otimes|\psi_j\rangle\langle\psi_j|,
\end{equation}
where
\begin{align}
 |\psi_0\rangle = \frac{1}{\sqrt{2}}\left(|0\rangle+|1\rangle\right), \quad |\psi_1\rangle = \frac{1}{\sqrt{2}}\left(|0\rangle-|1\rangle\right), \quad |\psi_2\rangle = |2\rangle.
\end{align}
It is not hard to calculate $M$:
\begin{equation}\label{eq:forExa_1}
M=
\begin{pmatrix}
1/2 & 1/2 & 0\\
1/2 & 1/2 & 0\\
0 & 0 & 1
\end{pmatrix}.
\end{equation}
which has rank 2. Define
\begin{equation}\label{eq:forExa_2}
P=
\begin{pmatrix}
4/9 & 0 & 0\\
0 & 0 & 2/9\\
0 & 2/9 & 1/9
\end{pmatrix}.
\end{equation}

It can be easily verified that if the two players measure the state in computational basis, the probability distribution yielded is uniform, which is a classical Nash equilibrium.

Now suppose that Player 1 uses a quantum computer. One can verify that the condition in Lemma
\ref{lem:advan} does not hold. For a concrete illustration, let us consider the protocol in Lemma
\ref{lem:advan} again. Player 1 first measures in the
basis $\{\ket{+}, \ket{-}, \ket{2}\}$. With probability 4/9, she observes $\ket{+}$, then changes it to $\ket{0}$. Player 2's state is also $\ket{+}$ in this case, thus a measurement in the computational basis gives the $\ket{00}$ and $\ket{01}$ each with half probability. Thus Player 1's payoff in this case is $2\cdot \frac{1}{2} - 1\cdot \frac{1}{2} = \frac{1}{2}$. The second case is that Player 1 observes $\ket{-}$, which happens with
probability 2/9, and Player 2's state is $\ket{2}$ for sure. Player 1 changes her part
to $\ket{2}$, and gets payoff $2$. The third case is that Player 1 observes $\ket{2}$, which happens with probability 1/3, leaving Player 2 $\sigma_3 = (2/3) \ket{1}\bra{1} + (1/3) \ket{2}\bra{2}$. Player 1 then changes her qubit to $\ket{1}$, collides with Player 2's outcome with probability 1/3, thus Player 1's payoff is $2\cdot \frac{1}{3} - 1\cdot \frac{2}{3} = 0$. On average, the quantum player has a payoff of $(4/9)(1/2) + (2/9)\cdot 2 + (1/3) \cdot 0 = 2/3.$

It should be pointed out that the matrix $P$ achieving the quantum
advantage of 2/3 is not unique. For example, the following matrix
also works with the same effect:
\begin{equation}\label{eq:forExa_3}
P=
\begin{pmatrix}
2/9 & 2/9 & 0\\
0 & 0 & 2/9\\
1/9 & 1/9 & 1/9
\end{pmatrix}.
\end{equation}

\subsection{Optimization}

In this subsection, we show that the 3-dimensional example in the
above subsection is actually optimal for $M$ defined in Eq.\eqref{eq:forExa_1}. Actually the theorem below shows more. Note that
if the rank of $M$ is 1, it is easy to prove that $M$ must be the
uniform matrix, and the quantum advantage must be zero, thus in the
following we suppose the rank of $M$ to be 2.

\begin{Thm}\label{thm:Opt}
Suppose that measuring the state $\rho$ gives a classical correlated
equilibrium. Suppose the columns of $M$ are $M^0,M^1$ and $M^2$.
Without loss of generality, suppose $M^0=xM^1+(1-x)M^2$, where
$0\leq x\leq1$. Then the quantum advantage
\begin{equation}
QA\leq\frac{1}{3}+\frac{1}{3x_b},
\end{equation}
where $x_b=\max\{x,1-x\}$.
\end{Thm}
\vspace{0.1in}
\begin{proof}
By Lemma \ref{lem:ce}, for any $0\leq k,l\leq 2$ we have
\begin{equation}\label{eq:forOpt_1}
\sum_{i,j=0}^{2}p(i,j)|\langle k|\psi_i\rangle|^2\cdot|\langle
l|\psi_j\rangle|^2=\frac{1}{9}.
\end{equation}
This turns out to be equivalent
to
\begin{equation}
M\cdot P\cdot M^T=\frac{J}{9}.
\end{equation}
Noting that $M\cdot (J/9)\cdot M^T=J/9$, we know that $P$ can be
expressed as
\begin{equation}\label{eq:forOpt_2}
P = \frac{J}{9}+\bar{P},
\end{equation}
where $M\cdot\bar{P}\cdot M^T=0$. By straightforward calculation,
one can show that Eq.\eqref{eq:forOpt_2} indicates
\begin{equation}
M\cdot\bar{P}=0.
\end{equation}
Considering the form of $M$, $\bar{P}$ can now be expressed as
\begin{equation}
\bar{P}=
\begin{pmatrix}
k_0 & k_1 & k_2\\
-k_0x & -k_1x & -k_2x\\
-k_0(1-x) & -k_1(1-x) & -k_2(1-x)
\end{pmatrix},
\end{equation}
where $k_0,k_1$ and $k_2$ are real numbers.

According to the discussion above, we know that the maximal quantum
advantage is
\begin{equation}
QA = \sum_{i=0}^{2}p_1(i)[2\cdot\langle
l_i|\sigma_i|l_i\rangle-1\cdot(1-\langle l_i|\sigma_i|l_i\rangle)],
\end{equation}
where $l_i=\max_{l}\langle l|\sigma_i|l\rangle$. Then it holds that
\begin{align*}
QA & =  3\sum_{i=0}^{2}p_1(i)\cdot\langle l_i|\sigma_i|l_i\rangle -1 \\
& = 3\sum_{i,j=0}^{2}p(i,j)|\langle l_i|\psi_j\rangle|^2 -1 \\
& = 3\sum_{i,j=0}^{2}\left(\frac{1}{9}+\bar{p}(i,j)\right)|\langle l_i|\psi_j\rangle|^2 -1 \\
& = 3\sum_{i=0}^{2}\left(\sum_{j=0}^{2}\bar{p}(i,j)|\langle
l_i|\psi_j\rangle|^2\right),
\end{align*}
where $\bar{p}(i,j)$ is the element of $\bar{P}$. At the same time,
it can be obtained that $l_i=\max_{l}\sum_j \bar{p}(i,j)|\langle
l|\psi_j\rangle|^2$. Besides, recall that the rank of $M$ is 2, then
there must be one row of $M$, say $M_2$, has the form of
$aM_0+(1-a)M_1$, where $M_0$ and $M_1$ are the other two rows of
$M$, and $0\leq a\leq1$. Then it can be known that every $l_i$ must
be $0$ or $1$. Based on the form of $\bar{P}$, we have that $l_0\neq
l_1=l_2$. Without loss of generality, we suppose $l_0=0$, and
$l_1=l_2=1$. Then
\begin{align*}
QA & =  3\sum_{j=0}^{2}\bar{p}(0,j)|\langle
0|\psi_j\rangle|^2+3\sum_{j=0}^{2}(\bar{p}(1,j)+\bar{p}(2,j))|\langle
1|\psi_j\rangle|^2 \\
& = 3\sum_{j=0}^{2}\bar{p}(0,j)|\langle
0|\psi_j\rangle|^2-3\sum_{j=0}^{2}\bar{p}(0,j)|\langle
1|\psi_j\rangle|^2.
\end{align*}
Note that $\bar{P}+J/9$ is a matrix with nonnegative elements.
Thus, for any $0\leq i\leq 2$, if $k_i\geq0$ we have
\begin{equation}\label{eq:forOpt_3}
-k_ix\geq-\frac{1}{9} \quad \text{and} \quad
-k_i(1-x)\geq-\frac{1}{9},
\end{equation}
and if $k_i<0$, we have $-k_i\leq\frac{1}{9}$. And
Eq.\eqref{eq:forOpt_3} indicates that if $0<x<1$,
\begin{equation}\label{eq:forOpt_4}
k_i\leq\frac{1}{9x} \quad \text{and} \quad k_i\leq\frac{1}{9(1-x)},
\end{equation}
which is equivalent to
\begin{equation}\label{eq:forOpt_5}
k_i\leq\frac{1}{9x_b}.
\end{equation}
Actually, Eq.\eqref{eq:forOpt_5} also holds when $x=0$ or $x=1$.
Therefore, we obtain that
\begin{align*}
QA & = 3\sum_{j=0}^{2}\bar{p}(0,j)|\langle
0|\psi_j\rangle|^2-3\sum_{j=0}^{2}\bar{p}(0,j)|\langle
1|\psi_j\rangle|^2 \\
& = 3\sum_{j=0}^{2}k_j|\langle
0|\psi_j\rangle|^2-3\sum_{j=0}^{2}k_j|\langle
1|\psi_j\rangle|^2 \\
& \leq 3\cdot\frac{1}{9x_b}+3\cdot\frac{1}{9} \\
& = \frac{1}{3}+\frac{1}{3x_b},
\end{align*}
where the relationship
$\sum_j|\langle0|\psi_j\rangle|^2=\sum_j|\langle1|\psi_j\rangle|^2=1$
is utilized.
\end{proof}

Go back to the example in the above subsection. Note that for $M$ in
Eq.\eqref{eq:forExa_1} we have $M^0=1\cdot M^1+0\cdot M^2$(thus in
order to utilize Theorem.\ref{thm:Opt}, we need to adjust the order
of the columns). Thus we can choose $x=0$, and then $x_b=1$. As a
result, the discussion above shows that $QA\leq2/3$, which means the
choice of $P$ in Eq.\eqref{eq:forExa_2} is optimal for $M$ in
Eq.\eqref{eq:forExa_1}.

\section{Conclusion}

In this paper, we have shown that in certain fair equilibria of some
natural games, the player with a quantum computer can be more
powerful than the classical opponent, even if the equilibria have
zero entanglement or discord. This indicates that at least in games,
the standard understanding that nonzero discord is necessary to
produce non-classical correlations needs to be adjusted. Our work
provides new visions for further studies on quantum information,
especially when at least two parties are involved with different
objectives.

From the mathematical perspective, some questions remain open. Two
of them are listed as below: (1) What is the maximum gain in a
zero-sum $[-1,1]$-normalized game\footnote{A game is
$[-1,1]$-normalized if all utility functions have ranges within
$[-1,1]$.} on a state in symmetric subspace without entanglement?
(2) What is the maximum gain in a zero-sum $[-1,1]$-normalized game
on a state in symmetric subspace without discord?

\section*{Acknowledgments} Z.W. thanks Leong Chuan
Kwek and Luming Duan for helpful comments. Z.W. was supported by the Singapore National Research Foundation under NRF RF Award
No.~NRF-NRFF2013-13 and the WBS grant under contract no.
R-710-000-007-271. S.Z. was supported by Research Grants Council of the Hong Kong S.A.R. (Project no. CUHK419011, CUHK419413), and this research benefited from a visit to Tsinghua University partially supported by China Basic
Research Grant 2011CBA00300 (sub-project 2011CBA00301).

\bibliographystyle{alpha}
\bibliography{QGame}
\end{document}